\documentclass[sigconf, nonacm]{acmart}
 
\usepackage{enumitem}
\usepackage{listings}
\usepackage[ruled,linesnumbered]{algorithm2e}
\usepackage{multirow}
\usepackage{bm}
\usepackage[misc]{ifsym}

\newcommand\vldbdoi{XX.XX/XXX.XX}
\newcommand\vldbpages{XXX-XXX}
\newcommand\vldbvolume{14}
\newcommand\vldbissue{1}
\newcommand\vldbyear{2020}
\newcommand\vldbauthors{\authors}
\newcommand\vldbtitle{\shorttitle} 
\newcommand\vldbavailabilityurl{URL_TO_YOUR_ARTIFACTS}
\newcommand\vldbpagestyle{plain} 

\newtheorem{myDef}{\noindent \textbf{Definition}}

\newtheorem{myLemma}{\noindent \textbf{Lemma}}

\begin{document}
\title{Efficient Data-aware Distance Comparison Operations for High-Dimensional Approximate Nearest Neighbor Search}



\author{Liwei Deng, Penghao Chen, Ximu Zeng}
\affiliation{%
  \institution{University of Electronic Science and Technology of China}
}
\email{deng_liwei, penghaochen, ximuzeng@std.uestc.edu.cn}

\author{Tianfu Wang}
\affiliation{%
  \institution{University of Science and Technology of China}
}
\email{tianfuwang@mail.ustc.edu.cn}

\author{Yan Zhao}
\affiliation{%
  \institution{University of Electronic Science and Technology of China\Letter}
}
\email{yanz@cs.aau.dk}

\author{Kai Zheng}
\affiliation{%
  \institution{University of Electronic Science and Technology of China\Letter}
}
\email{zhengkai@uestc.edu.cn}
\authornote{Liwei Deng and Penghao Chen are equally contributed to this work. \Letter \ Kai Zheng is with the School of Computer Science and Engineering, University of Electronic Science and Technology of China. Yan Zhao and Kai Zheng are with Shenzhen Institute for Advanced Study, University of Electronic Science and Technology of China.}

\begin{abstract}
High-dimensional approximate $K$ nearest neighbor search (AKNN) is a fundamental task for various applications, including information retrieval.
Most existing algorithms for AKNN can be decomposed into two main components, i.e., candidate generation and distance comparison operations (DCOs). While different methods have unique ways of generating candidates, they all share the same DCO process. 
In this study, we focus on accelerating the process of DCOs that dominates the time cost in most existing AKNN algorithms.
To achieve this, we propose an \underline{D}ata-\underline{A}ware \underline{D}istance \underline{E}stimation approach, called \emph{DADE}, 
which approximates the \emph{exact} distance in a lower-dimensional space.
We theoretically prove that the distance estimation in \emph{DADE} is \emph{unbiased} in terms of data distribution. Furthermore, we propose an optimized estimation based on the unbiased distance estimation formulation.
In addition, we propose a hypothesis testing approach to adaptively determine the number of dimensions needed to estimate the \emph{exact} distance with sufficient confidence. 
We integrate \emph{DADE} into widely-used AKNN search algorithms, e.g., \emph{IVF} and \emph{HNSW}, and conduct extensive experiments to demonstrate the superiority.
\end{abstract}

\maketitle

\vspace{-0.35cm}
\pagestyle{\vldbpagestyle}
\begingroup\small\noindent\raggedright\textbf{PVLDB Reference Format:}\\
\vldbauthors. \vldbtitle. PVLDB, \vldbvolume(\vldbissue): \vldbpages, \vldbyear.
\endgroup
\begingroup
\renewcommand\thefootnote{}\footnote{\noindent
This work is licensed under the Creative Commons BY-NC-ND 4.0 International License. Visit \url{https://creativecommons.org/licenses/by-nc-nd/4.0/} to view a copy of this license. For any use beyond those covered by this license, obtain permission by emailing \href{mailto:info@vldb.org}{info@vldb.org}. Copyright is held by the owner/author(s). Publication rights licensed to the VLDB Endowment. \\
\raggedright Proceedings of the VLDB Endowment, Vol. \vldbvolume, No. \vldbissue\ %
ISSN 2150-8097. \\
\href{https://doi.org/\vldbdoi}{doi:\vldbdoi} \\
}\addtocounter{footnote}{-1}\endgroup

\vspace{-0.35cm}
\ifdefempty{\vldbavailabilityurl}{}{
\vspace{.3cm}
\begingroup\small\noindent\raggedright\textbf{PVLDB Artifact Availability:}\\
The source code, data, and/or other artifacts have been made available at \url{https://github.com/Ur-Eine/DADE}.
\endgroup
}

\vspace{-0.3cm}
\section{Introduction}

Searching for the K nearest neighbors (KNN) in the high-dimensional Euclidean space is pivotal for various applications, such as data mining~\cite{cover1967nearest,wang2019mcne,deng2022efficient,deng2024learning}, information retrieval~\cite{grbovic2018real,huang2020embedding}, scientific computing, recommendation systems~\cite{huang2020embedding,okura2017embedding,lian2020geography}, and large language models~\cite{chen2024large,asai2023tutorial,li2023skillgpt}. However, with the increase of dimensions, the performance of existing indexing algorithms such as R-Tree~\cite{guttman1984r,beckmann1990r} and KD-Tree~\cite{bentley1975multidimensional,wald2006building} for exact searching degrades to that of brute-force search, which is time-consuming and 
often fails to meet the practical application requirements. This phenomenon is known as the curse of dimensionality. Thus, to achieve the tradeoff between accuracy and latency, existing research mainly focuses on developing algorithms for approximate K nearest neighbors (AKNN) search, which aim to provide neighbors with acceptable recall 
while significantly improving search speed.

Existing studies for AKNN can be broadly categorized into four main approaches: (1) graph-based~\cite{zhao2023towards,malkov2018efficient,fu12fast,iwasaki2018optimization,munoz2019hierarchical,chen2023finger,wang2024starling}, (2) quantization-based~\cite{jegou2010product,ge2013optimized,babenko2014inverted,gao2024rabitq,kalantidis2014locally,liu2017pqbf}, (3) tree-based~\cite{beygelzimer2006cover,ciaccia1997m,dasgupta2008random,muja2014scalable}, and (4) hashing-based~\cite{datar2004locality,gan2012locality,zheng2020pm,tian2023db}. Despite 
their different methodologies, 
these algorithms generally follow a common paradigm: 
they first generate 
candidate neighbors 
and then refine 
these candidates to identify the K nearest neighbors. It should be noted that these algorithms share a similar process in the refinement process, which involves
maintaining a max-heap $\mathcal{Q}$ to store the current KNNs. 
For each new candidate, they check whether its distance to the query is less than the maximum in $\mathcal{Q}$~\cite{gao2023high}. If the distance is not less than the maximum, the candidate is discarded. Otherwise, the max-heap $\mathcal{Q}$ is updated to remove the object with the maximum distance and add the new candidate. This checking process is called \emph{distance comparison operation} (DCO)~\cite{gao2023high}. 
{\color{black}{DCO is widely adopted in AKNN approaches. For example in HNSW and IVF, they first compute the \emph{exact} distance between the query and a candidate, and then compare this distance to the maximum in Q.}}
The 
computation in the full-dimensional space has a time complexity of $O(D)$, where $D$ is the number of dimensions in the Euclidean space. 
{\color{black}{Gao et al.~\cite{gao2023high} demonstrate that the cost of performing DCOs dominates the total query time of \emph{HNSW}.}}
For example, on the DEEP dataset with 256 dimensions, DCOs take 77.2\% of the total running time.

An intuitive idea to speed up the process of DCOs is to determine the distance between query and candidate without the calculation of the exact distance. This goal can be achieved through some distance approximation methods such as product quantization~\cite{jegou2010product,ge2013optimized}. 
{\color{black}{However, despite these methods can enhance the efficiency of distance computation, the accuracy will degrade dramatically. }}
Thus, these methods are more suitable for generating candidates rather than obtaining KNN from the generated candidates~\cite{wang2017survey}. Recently, Gao et al.~\cite{gao2023high} systematically study solutions for DCOs and propose a method called $ADSampling$, which firstly applies a \emph{random orthogonal transformation}~\cite{choromanski2017unreasonable} to the original vectors and then adaptively determines the number of dimensions to be sampled for each object during the query phase based on the DCO for distance estimation. They provide a theoretical proof to demonstrate that the distance estimation of $ADSampling$ is \emph{unbiased} in terms of the random transformation and that 
the failure probability is bounded by a constant.
However, $ADSampling$ is data-oblivious and cannot provide an accurate distance estimation for a specific dataset, which hinders it from being optimal.


In this study, we propose a new method called \emph{DADE} for \underline{D}ata-\underline{A}ware \underline{D}istance estimation~\footnote{{\color{black}{By distances, we refer to the Euclidean distance without further specification.}}}. Specifically, \emph{DADE} utilizes an \emph{orthogonal transformation} to rotate the original space and performs DCOs in the projected space. 
{\color{black}{The number of dimensions $d$ ($d\le D$) is adaptively determined to estimate the distance between the query and individual candidates.}}
Different from previous distance estimation methods that project all objects with \emph{equal} dimensions, \emph{DADE} is more similar to $ADSampling$ 
in that it employs a different number of dimensions for DCOs on different objects.
However, \emph{DADE} improves upon $ADSampling$ by deriving the orthogonal transformation based on data distribution rather than \emph{randomly}, leading to more accurate distance estimation within the same running time.
We 
provide theoretical proofs for two key points: (1) our distance estimation is \emph{unbiased} in terms of data distribution 
regardless of the number of dimensions, and (2) 
{\color{black}{compared to other \emph{unbiased} distance estimation methods, our approach is optimized in terms of variance. These proofs demonstrate that \emph{DADE} is a better distance estimation method when the transformation is \emph{orthogonal}.}}
Furthermore, to determine when to expand the number of dimensions, a hypothesis testing approach is adopted, in which 
the significance level is controlled by a probability defined and empirically derived from the data objects.







In summary, we conclude our contributions as follows.

\begin{itemize}[leftmargin=*]
\item We propose a new method \emph{DADE}, which can be 
integrated as a plug-in component to 
accelerate the search process in existing AKNN algorithms such as \emph{IVF} and \emph{HNSW}.

\item {\color{black}{We 
provide a theoretical proof showing that the proposed distance estimation is \emph{unbiased} and \emph{optimized} in terms of data distribution when the transformation is \emph{orthogonal}. }}

\item We propose a hypothesis testing method to \emph{adaptively} control the number of dimensions used in distance estimation. The probability of estimation deviation is empirically approximated from the data objects, addressing the challenge of explicitly expressing the data distribution. 

\item We conduct extensive experiments on real datasets to show the superiority of our method. For example, on the DEEP dataset, \emph{DADE} improves the queries per second (QPS) by over $40\%$ on \emph{HNSW} compared with the state-of-the-art approach $ADSampling$, while maintaining the same level of accuracy. 


\end{itemize}

\section{Preliminaries}
We proceed to present the necessary preliminaries and then define the problem addressed. 


\begin{myLemma}
\label{lemma:zeromean}
{\color{black}{Assume $X_1,X_2 \in \mathbb{R}^D$ are independent and identically distributed random vectors, and $Y_i=X_i-\mathbb{E}[X_i]\ (for\ i=1,2)$. For $W_D \in \mathbb{R}^{D\times D}, {W_D}^T W_D=\mathbf{I}$, the following equation holds:}}
\begin{equation}
\small
\label{eq:eq1}
\begin{aligned}
\|{W_D}^T Y_1 - {W_D}^T Y_2\|_2^2 & = \|{W_D}^T X_1 - {W_D}^T X_2\|_2^2
\end{aligned}
\end{equation}
\end{myLemma}
\begin{proof}
\vspace{-0.1cm}
\small
{\color{black}
$ \|{W_D}^T Y_1 - {W_D}^T Y_2\|_2^2 = \|{W_D}^T (Y_1 - Y_2)\|_2^2 \\
 = \|{W_D}^T [(X_1 - \mathbb{E}[X_1]) - (X_2 - \mathbb{E}[X_2])]\|_2^2 \\
 = \|{W_D}^T (X_1 - X_2)\|_2^2 = \|{W_D}^T X_1 - {W_D}^T X_2\|_2^2 $.
 }
\vspace{-0.2cm}
\end{proof}

Thus, without loss of generalization, we assume that the random vector in the subsequent proof is zero mean, i.e., $\mathbb{E}[X]=0$. The vectors used in this paper are all column vectors. In addition, we also provide a simple Lemma about \emph{orthogonal projection}, which will be used as a proof part in the following sections. 

\begin{myLemma}
\label{lemma:lemma_2}
Orthogonal projection does not change the sum of variances of all dimensions.
\end{myLemma}
\begin{proof}
\vspace{-0.1cm}
\small
$ \sum_{k=1}^D \text{Var}(x_k) = \mathbb{E}[X^T X] - \mathbb{E}[X]^T \mathbb{E}[X] \\
= \mathbb{E}[X^T W_D {W_D}^T X] = \sum_{k=1}^D \text{Var}({w_k}^T X)$. 

\noindent {\color{black}{Where $x_k$ is the component of the zero-mean random vector $X$ in the k-th dimension.}}
\vspace{-0.2cm}
\end{proof}

\begin{myDef}[Distance Comparison Operation~\cite{gao2023high}]
\label{def:dcos}
Given a query $\bm{q}$, an object $\bm{o}$ and a distance threshold $r$, the \textbf{distance comparison operation} (DCO) is to decide whether object $\bm{o}$ has its distance $dis$ to $\bm{q}$ no greater than $r$ and if so, return $dis$. 
\end{myDef}


\section{Methodology}
We develop a new method called \emph{DADE} to perform DCOs with better efficiency than existing approaches. Specifically, \emph{DADE} first rotates the original Euclidean space with a data-aware orthogonal transformation and then conducts DCOs based on the projected space, in which the distance between a query and a candidate is estimated in a subspace with \emph{fewer} dimensions for better efficiency. Compared with existing studies, \emph{DADE} shows three main differences: (1) {\color{black}{Compared with random projection methods~\cite{datar2004locality,gan2012locality,sun2014srs} that project objects into vectors with \emph{equal} dimensions, which may be knowledge demanding and difficult to set in practice~\cite{gao2023high}, \emph{DADE} estimates the distance with \emph{adaptive} dimensions.}} (2) Compared with the SOTA method, i.e., $ADSampling$~\cite{gao2023high}, the transformation in \emph{DADE} is data-aware, which provides a more accurate distance approximation with the same level of running time. (3) The number of dimensions to be used for distance computation is determined by hypothesis testing, in which the \emph{unknown} data distribution for the hypothesis testing is empirically approximated. 
The details of \emph{DADE} are elaborated in the following sections.



\subsection{Unbiased Estimation}
An intuitive idea to speed up the process of DCOs is to determine whether $dis \le r$ without computing the exact distance between the query and the candidate. Thus, a set of distance approximation methods can be adopted. For example, $ADSampling$ conducts \emph{randomly orthogonal transformation} on an object with a random matrix $W_d \in \mathbb{R}^{D\times d}$ where $d\le D$, and then perform the DCOs with the approximate distance.
However, these methods are data-oblivious, which prevents them from optimum. To deal with this, we first provide a data-aware \emph{unbiased} estimation as the following Lemma. 

\begin{myLemma}
Given a set of orthogonal bases in $\mathbb{R}^{D}$ Euclidean space, i.e., $W_d := [w_1, w_2, \cdots, w_d] \in \mathbb{R}^{D\times d}$, where $\forall i\ne j$, ${w_i}^T w_j = 0$ and ${w_i}^T w_i = 1$. Assume $X, X_1, and \ X_2 \in \mathbb{R}^{D}$ are three independent and identically distributed random vectors. The following equation holds:
\begin{equation}
\small
\begin{aligned}
\mathbb{E}[\|X_1 - X_2\|_2^2] &= \frac{\sigma^2(1,D)}{\sigma^2(1,d)} \mathbb{E}[\|{W_d}^T X_1 - {W_d}^T X_2\|_2^2]
\end{aligned}
\end{equation}
\noindent where $\sigma^2(i,j)= \sum_{k=i}^j \text{Var}({w_k}^T X), 1 \leq i \leq j \leq D$ and $\text{Var}({w_k}^T X)$ indicates the variance of ${w_k}^T X$.
\end{myLemma}
\begin{proof}
For $\forall d \in \{1,2,...,D\}$, we have:
\begin{equation}
\small
\label{eq:eq2}
\begin{aligned}
& \mathbb{E}[\|{W_d}^T X_1 - {W_d}^T X_2\|_2^2] = \mathbb{E}[[{W_d}^T (X_1 - X_2)]^T [{W_d}^T (X_1 - X_2)]] \\
&= \mathbb{E}[\text{tr}((X_1 - X_2)^T W_d {W_d}^T (X_1 - X_2))] \\
&= \mathbb{E}[\text{tr}({W_d}^T (X_1 - X_2) (X_1 - X_2)^T W_d)] \\
&= \text{tr}({W_d}^T (\mathbb{E}[X_1 {X_1}^T] + \mathbb{E}[X_2 {X_2}^T] - \mathbb{E}[X_1 {X_2}^T] - \mathbb{E}[X_2 {X_1}^T]) W_d) \\
&= \text{tr}({W_d}^T (2\mathbb{E}[X X^T] - \mathbb{E}[X_1] \mathbb{E}[{X_2}^T] - \mathbb{E}[X_2] \mathbb{E}[{X_1}^T]) W_d) \\
&= 2\mathbb{E}[\text{tr}({W_d}^T X X^T W_d)] {\color{black}{ = 2\mathbb{E}[X^T W_d {W_d}^T X]}} = 2\mathbb{E}\left[\sum_{k=1}^d (X^T w_k {w_k}^T X) \right] \\
&= 2\sum_{k=1}^d \mathbb{E}[({w_k}^T X)^2] = 2\sum_{k=1}^d \text{Var}({w_k}^T X) = 2\sigma^2(1,d)
\end{aligned}
\end{equation}


{\color{black}{From line 1 to line 2, the inner part of the expectation is a real number. Thus, the trace operation can be safely added. From line 2 to line 3, the orders among elements are changed thanks to the nature of trace operation. From line 3 to line 4, the order of the experation and trace is changed, where the correctness stems form the addictivity of mathematical expectations. Line 5 is obtained since all $\mathbb{E}[X_i]$ are equal to 0 (see the zero-mean assumption in \textbf{Lemma~\ref{lemma:zeromean}}). }}

When $d=D$, we have $\mathbb{E}[\|X_1 - X_2\|_2^2] = \mathbb{E}[\|{W_D}^T (X_1 - X_2)\|_2^2]$ due to the property of orthogonal projection. Thus, we have the following equation.
\begin{equation}
\label{eq:unbiased_estimation}
\begin{aligned}
\mathbb{E}[\|X_1 - X_2\|_2^2] &= \mathbb{E}[\|{W_D}^T (X_1 - X_2)\|_2^2] \\
&= \frac{\sigma^2(1,D)}{\sigma^2(1,d)} \mathbb{E}[\|{W_d}^T X_1 - {W_d}^T X_2\|_2^2]
\end{aligned}
\end{equation}
\vspace{-0.3cm}
\end{proof}

From above Lemma, we can see that the \emph{exact} distance can be \emph{unbiasedly} estimated by the variance of each dimension and the distance in projected $\mathbb{R}^{d}$ Euclidean space. 

\subsection{{\color{black}{Optimized Estimation}}}
{\color{black}{Equation~\ref{eq:unbiased_estimation} provides a general formulation of the unbiased estimation of Euclidean distance in terms of data distribution. In this subsection, we further study the optimal estimation, which minimizes the variance of the difference between the estimated distance and the true distance.}}
We let $\Delta X := X_1 - X_2 \in \mathbb{R}^{D}$. 
{\color{black}{For any $d$, the optimal estimation can be described as an optimization problem:}}
\begin{flalign}
\label{eq:optimization}
\begin{split}
\min_{W_D \in \mathbb{R}^{D \times D}} \quad & \mathbb{E}\left[\left(\frac{\sigma^2(1,D)}{\sigma^2(1,d)} \|{W_d}^T \Delta X\|_2^2 - \|{W_D}^T \Delta X\|_2^2 \right)^2 \right] \\
\text{s.t.} \quad & {W_D}^T W_D = \mathbf{I}
\end{split}
\end{flalign}


{\color{black}{It should be noted that we should only have a uniform $W_D$, where $W_d$ is the first $d$ dimension of $W_D$, since obtaining different $W_d$ for each $d$ is time-consuming to perform transformation and memory-consuming to store such plenty of transformed data objects. Thus, we propose to optimize an alternative objective as described as follows.}}


\begin{myLemma}
{\color{black}{For any $d$, minimizing Equation~\ref{eq:optimization} can be approximately achieved through maximizing $\sigma^2(1,d)$.}}
\end{myLemma}
\begin{proof}
For the internal part, we have:
\begin{equation}
\small
\label{eq:lemma3_eq1}
\begin{aligned}
& \left(\frac{\sigma^2(1,D)}{\sigma^2(1,d)} \|{W_d}^T \Delta X\|_2^2 - \|{W}_D^T \Delta X\|_2^2 \right)^2 \\
& = \left[ \frac{\sigma^2(1,D)}{\sigma^2(1,d)} \sum_{k=1}^d \left(\Delta X^T w_k {w_k}^T \Delta X \right) - \sum_{k=1}^D \left(\Delta X^T w_k {w_k}^T \Delta X \right)\right]^2 \\
& = \left\{ \sum_{k=1}^D \left[ \left( \frac{\sigma^2(d+1,D)}{\sigma^2(1,d)} \mathbb{I}_{k \leq d} - \mathbb{I}_{k > d}\right) \Delta X^T w_k {w_k}^T \Delta X\right] \right\}^2
\end{aligned}
\end{equation}
\noindent where $\mathbb{I}_{k\le d}$ equals to $1$ if $k\le d$, otherwise $0$. We define $L_d\ \in \mathbb{R}^{D\times D}$ as a diagonal matrix as follows:
\begin{equation}
\small
L_d := \text{diag}\left(\frac{\sigma^2(d+1,D)}{\sigma^2(1,d)},...,\frac{\sigma^2(d+1,D)}{\sigma^2(1,d)},-1,...,-1 \right)
\end{equation}
\noindent where the first $d$ elements of $L_d$ are $\sigma^2(d+1,D)/\sigma^2(1,d)$, and the remaining in the main diagonal elements are $-1$. We also define {\color{black}{$\Sigma := {W_D}^T \Delta X \Delta X^T W_D \ \in \mathbb{R}^{D\times D}$}}. With these definitions, Equation~\ref{eq:lemma3_eq1} can be rewritten in matrix form as follows:
\begin{equation}
\small
\begin{aligned}
& \left\{ \sum_{k=1}^D \left[ \left( \frac{\sigma^2(d+1,D)}{\sigma^2(1,d)} \mathbb{I}_{k \leq d} - \mathbb{I}_{k > d} \right) \Delta X^T w_k {w_k}^T \Delta X \right] \right\}^2 \\
& = (\Delta X^T W_D L_d {W_D}^T \Delta X)^2 = \Delta X^T W_D L_d {W_D}^T \Delta X \Delta X^T W_D L_d {W_D}^T \Delta X \\
& = \text{tr}(\Delta X^T W_D L_d {W_D}^T \Delta X \Delta X^T W_D L_d {W_D}^T \Delta X) \\
& = \text{tr}(({W_D}^T \Delta X) (\Delta X^T W_D L_d {W_D}^T \Delta X \Delta X^T W_D L_d)) \\
& {\color{black}{ = \text{tr}(\Sigma L_d \Sigma L_d) \le \text{tr}\left(\Sigma {L_d}^{\frac{1}{2}} ({L_d}^{\frac{1}{2}})^H \Sigma {L_d}^{\frac{1}{2}} ({L_d}^{\frac{1}{2}})^H \right) }} \\
& {\color{black}{ = \text{tr}\left(({L_d}^{\frac{1}{2}})^H \Sigma {L_d}^{\frac{1}{2}} ({L_d}^{\frac{1}{2}})^H \Sigma {L_d}^{\frac{1}{2}} \right) = \| ({L_d}^{\frac{1}{2}})^H \Sigma {L_d}^{\frac{1}{2}} \|_F^2}} \\
& {\color{black}{ \le \| \Sigma \|_F^2 \| {L_d}^{\frac{1}{2}} \|_F^4 = \| {W_D}^T \Delta X \Delta X^T W_D \|_F^2 \| {L_d}^{\frac{1}{2}} \|_F^4 = \| \Delta X \Delta X^T \|_F^2 \| {L_d}^{\frac{1}{2}} \|_F^4}}
\end{aligned}
\end{equation}
{\color{black}{Thus, our optimization goal (c.f. Equation~\ref{eq:optimization}) can be approximately presented as follows.}}
\begin{equation}
\small
\label{eq:rewritten}
\begin{aligned}
{\color{black}{\min_{W_D \in \mathbb{R}^{D \times D}}}}{\color{black}{ \mathbb{E}[\| \Delta X \Delta X^T \|_F^2 \| {L_d}^{\frac{1}{2}} \|_F^4] = \mathbb{E}[\| \Delta X \Delta X^T \|_F^2 ] \| {L_d}^{\frac{1}{2}} \|_F^4 }}
\end{aligned}
\end{equation}
\noindent 
{\color{black}{From this equation, we can know that this objective can be achieved through minimizing $\| {L_d}^{\frac{1}{2}} \|_F^4$. According to the definition of $L_d$, the final optimization is to minimize $\sigma^2(d+1,D)/\sigma^2(1,d)$. From \textbf{Lemma~\ref{lemma:lemma_2}}, we can know that orthogonal projection does not change the sum of variances of all dimensions, which means $\sigma^2(d+1, D) + \sigma^2(1,d)$ is a constant. Therefore, Equation~\ref{eq:rewritten} is equivalent to maximize $\sigma^2(1,d)$. Proof complete. }} 
\end{proof}

Thus, from this Lemma, by the definition of $\sigma^2(1,d)$, our optimization goal can be reformulated as follows:
\begin{equation}
\small
\label{eq10}
\begin{split}
\max_{W_D \in \mathbb{R}^{D\times D}}\sigma^2(1,d) &= \sum_{k=i}^d \text{Var}({w_k}^T X) = \mathbb{E}[X^T W_d {W_d}^T X] \\
&= \max \text{tr}({W_d}^T \mathbb{E}[X X^T] W_d) \\ 
\text{s.t.} \ \  & {W_d}^T W_d = \mathbf{I}
\end{split}
\end{equation}

According to the previous study~\cite{zhou2021machine}, it is easy to know that this is the optimization objective of principal components analysis (PCA), in which $\mathbb{E}[X X^T]$ is approximated by all data objects.
Therefore, the solution for this problem can be obtained through matrix decomposition on $\mathbb{E}[X X^T]$, where $\lambda_k$ is the $k$ \emph{largest} eigenvalue and $w_k$ is the corresponding eigenvector. 
\begin{equation}
\small
\label{eq11}
\begin{split}
    \mathbb{E}[X X^T] w_k = \lambda_k w_k,\ for\ k=1,2,...,d 
\end{split}
\end{equation}

Moreover, we can also notice the following equation hold: 
\begin{equation}
\small
\label{eq:eigenvalue}
\begin{split}
\text{Var}({w_k}^T X) = \mathbb{E}[{w_k}^T X X^T w_k] = {w_k}^T \mathbb{E}[X X^T] w_k = \lambda_k
\end{split}
\end{equation}

{\color{black}{Hence, combining Equation~\ref{eq:unbiased_estimation} and~\ref{eq:eigenvalue}, we can know that the optimized distance estimation is as follows}}:
\begin{equation}
\small
\label{eq:pca_dis_estimate}
\begin{aligned}
\mathbb{E}[\|X_1 - X_2\|_2^2] &= \frac{\sum_{k=1}^D \lambda_k}{\sum_{k=1}^d \lambda_k} \mathbb{E}[\|{W_d}^T X_1 - {W_d}^T X_2\|_2^2]
\end{aligned}
\end{equation}

{\color{black}{It is worth mentioning that for any $d$, the transformation $W_d$ obtained through PCA is optimal for Equation~\ref{eq10}. Thus, we only need to transform once and store the transformed data objects once to make the time and space consumption acceptable. We provide an empirical study on the DEEP dataset to show the variance of each dimension in the projected space by comparing \emph{PCA} and \emph{randomly orthogonal projection} (ROP) as shown in the left panel of Figure~\ref{fig:running_example}.}}
We can see that the random approach has almost uniform variance while \emph{PCA} can achieve greater variance with fewer dimensions. Thus, from Equation~\ref{eq:pca_dis_estimate}, we can know that \emph{PCA} is a more powerful approach to approximate the exact distance. It should be noted that the distance estimation of $ADSampling$~\cite{gao2023high}
is \emph{unbiased} in terms of the transformation rather than the data distribution. 

\begin{figure}
  \centering
  \includegraphics[width=\linewidth]{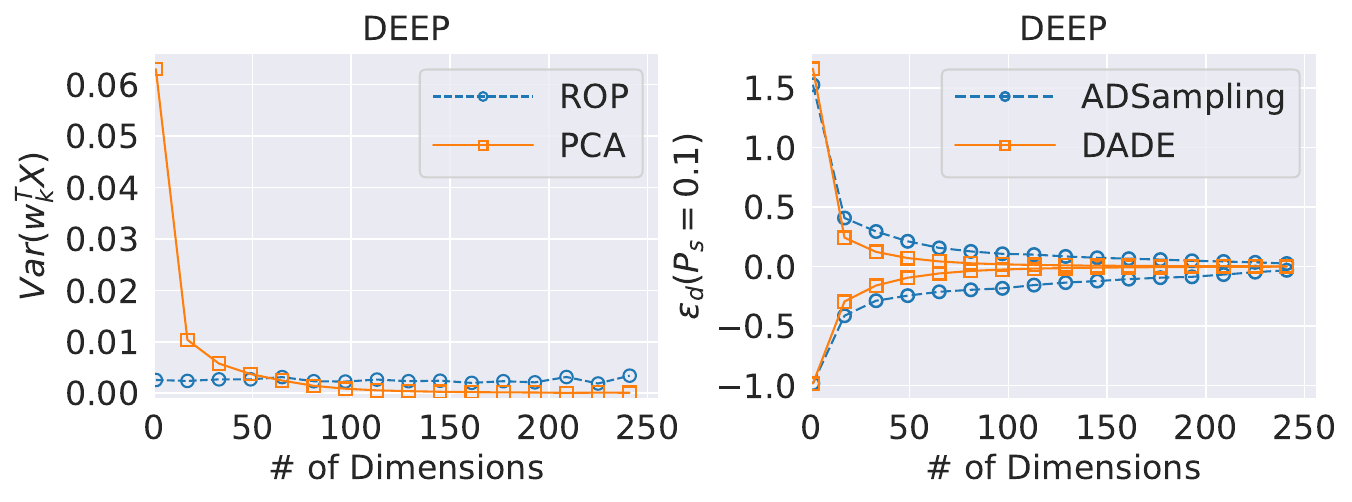}
  \caption{{\color{black}{Running Example on DEEP.}}}
  \label{fig:running_example}
\vspace{-0.3cm}
\end{figure}

\begin{algorithm}[t]
\caption{\emph{DADE}}
\label{alg:dade}
\KwIn{A transformed data vector $o'$, a transformed query vector $q'$, a distance threshold $r$ and the incremental step size $\Delta_d$}
\KwOut{The results of DCO (i.e., whether $dis < r$): $1$ means yes and $0$ means no; When the answer is yes, the exact distance is also returned}
Initialize the number of sampled dimensions $d$ to be $0$\; \label{alg:line1}
\While{$d < D$}{
$d=d+\Delta_d$\; \label{alg:line3}
{\color{black}{Using the first $d$ dimensions to compute the estimated distance $dis'$ according to Equation~\ref{eq:pca_dis_estimate}\;}} \label{alg:line4}
{\color{black}{Conduct the hypothesis testing as stated in section~\ref{sec:hypothesis}\;}} \label{alg:line5}
\If{$H_0$ is rejected and $d<D$}{ \label{alg:line6}
\Return 0\; \label{alg:line7}
}
\ElseIf{$H_0$ is not rejected and $d<D$}{ \label{alg:line9}
\textbf{Continue}\; \label{alg:line10}
}
\Else{
\Return $1$ and $dis'$ if $dis'\le r$ and $0$ otherwise\; \label{alg:line13}
}
}
\end{algorithm}
\vspace{-0.3cm}

\subsection{Dimension Expansion with Hypothesis Testing}
\label{sec:hypothesis}
Projecting objects into vectors with \emph{equal} dimensions to approximate the \emph{exact} distance usually has two issues. First, it is knowledge-demanding and difficult to determine the number of dimensions to approximate the \emph{exact} distance with sufficient confidence in practice. Second, conducting DCOs with \emph{equal} dimensions has inferior accuracy since different objects may require different numbers of dimensions to make a good decision (see Section~\ref{sec:feasibility}). Therefore, we leverage hypothesis testing to \emph{adaptively} determine the number of dimensions when the query object $q$, the candidate object $o$, and the distance of the $K$-th nearest neighbor are given. 

Specifically, we propose to determine the number of dimensions of distance estimation in an \emph{incremental} manner, which is similar to the method $ADSampling$. However, different from $ADSampling$ that has \emph{unbiased} estimation in terms of the transformation matrix, in which there exists a concentration inequality on the approximate distance (see Lemma 3 in~\cite{gao2023high} for details), our method \emph{DADE} provides an \emph{unbiased} estimation in terms of data distribution, where the data distribution is unknown and has unclear expression. Thus, it is difficult to set the significance level in the hypothesis testing. To deal with this problem, we define the following probability.
\begin{equation}
\small
\begin{aligned}
\mathbb{P}\{\frac{\|\frac{\sum_{k=1^D} \lambda_k}{\sum_{k=1^d} \lambda_k} {W_d}^T (X_1 - X_2)\|}{\|X_1 - X_2\|} - 1 > \epsilon_d \} = P_s
\end{aligned}
\end{equation}
\noindent where $P_s$ is the significance level to be set as a hyper-parameters. It means a probability that the difference between the approximated distance and the exact distance is greater than $\epsilon_d$. With the provided $P_s$, $\epsilon_d$ can be estimated through uniformly sampled data objects. It should be noted that $\epsilon_d$ may be different for a fixed $P_s$ and different $d$. Then, we define $dis'=\|\frac{\sum_{k=1^D} \lambda_k}{\sum_{k=1^d} \lambda_k} {W_d}^T (X_1 - X_2)\|$ as the estimated distance and $dis=\|X_1 - X_2\|$ as the exact distance for presentation convenience. The hypothesis testing can be conducted as follows. 

\noindent {\color{black}{(1) We define the null hypothesis $H_0: dis < r$ and its counterpart $H_1: dis \ge r$.}}

\noindent (2) We set the significance level $P_s$ empirically as a small value (e.g., $0.1$ in our experiments), which indicates that the difference between $dis'$ and $dis$ is bounded by $\epsilon_d \cdot dis$ with the failure probability at most $P_s$. 

\noindent (3) We check whether the event $dis' > (1 + \epsilon_d) \cdot r$ happens. If so, we can reject $H_0$ and conclude $H_1: dis > r$ with sufficient confidence since this event has a small probability, which is almost impossible to happen in one experiment. 

We also provide an empirical study on DEEP dataset as shown in the right part of Figure~\ref{fig:running_example}, in which the $x$-axis indicates the number of dimensions and $y$-axis presents $\epsilon_d$ where $\mathbb{P}(dis'/dis - 1 > \epsilon_d) = 0.1$ for the upper two curves and $\mathbb{P}(dis'/dis - 1 < \epsilon_d) = 0.1$ for the bottom two curves. From this figure, we can have two observations: (1) \emph{PCA} has a better approximation to the exact distance since it has smaller deviations with the same number of dimensions; (2) compared with \emph{random orthogonal projection}, \emph{PCA} needs smaller dimension to reach the same significance level for the estimated distance, which means \emph{DADE} is more efficient since it is more likely to reject $H_0$ compared with $ADSampling$ when the significance level is fixed. 

\subsection{\emph{DADE} Summarization}
The process of \emph{DADE} is summarized in Algorithm~\ref{alg:dade}, which takes the transformed data and query vectors, a distance threshold $r$ (i.e., the distance between the query and the $K$-th nearest neighbor), and $\Delta_d$ as inputs. It runs in an \emph{incremental} way, which initializes the number of dimensions as $0$ and increment it with $\Delta_d$ (Lines~\ref{alg:line1}-\ref{alg:line3}). At each loop, we first calculate the estimated distance $dis'$ (Line~\ref{alg:line4}) and then conduct the hypothesis testing (Line~\ref{alg:line5}). {\color{black}{If $H_0$ is rejected, we can conclude that $dis > r$ with sufficient confidence and exit the DCOs program immediately (Lines~\ref{alg:line6}-\ref{alg:line7}).}} If $H_0$ is not rejected, it means that we do not have enough confidence to judge whether $dis < r$. Thus, we have to continue to increment the number of dimensions to obtain more accurate $dis'$ (Lines~\ref{alg:line9}-\ref{alg:line10}). For the other situation (i.e., $d=D$), the $dis'$ will be the exact distance. Thus, we will directly compare $dis'$ and $r$ and return the results (Line~\ref{alg:line13}). 


\noindent \textbf{Failure Probability Analysis.} 
We provide the following Lemma to present the failure probability when \emph{DADE} is adopted.

\begin{myLemma}
The failure probability of \emph{DADE} is given by 
\begin{equation}
\small
\begin{aligned}
\mathbb{P}\{failure\} &= 0 \ if \ dis > r \\
\mathbb{P}\{failure\} &{\color{black}{\le \lfloor \frac{D-1}{\Delta d} \rfloor \cdot P_s \ if \ dis \le r}}
\end{aligned}
\end{equation}
\end{myLemma}
\begin{proof}
From Algorithm~\ref{alg:dade}, it is known that \emph{DADE} exits when $dis'>r$ or $d=D$. If $dis>r$, in these situations, our proposed \emph{DADE} returns $0$, which is always correct. Thus, $\mathbb{P}\{failure\} = 0$ can be concluded when $dis>r$. Now we consider the other situation that $dis \le r$, which can be verified as follows.
\begin{equation}
\small
\begin{aligned}
& \mathbb{P}\{failure\} = \mathbb{P}\{\exists d < D, dis' > (1 + \epsilon_d) \cdot r\} \\
& {\color{black}{\le \sum_{d=1}^{\lfloor (D-1)/\Delta d \rfloor} \mathbb{P}\{dis' > (\epsilon_{d} + 1) \cdot dis\} \le \sum_{d=1}^{\lfloor (D-1)/\Delta d \rfloor} P_s}}
\end{aligned}
\end{equation}
\noindent where the first equation holds since a failure happens if and only if we reject $H_0$ for some $d<D$. 
\end{proof}



Recall that DCO is ubiquitous in almost all AKNN algorithms. For example, for a graph-based method such as HNSW, \emph{greedy beam search}~\cite{wang2021comprehensive} is conducted at layer $0$, which is also adopted by most graph-based approaches~\cite{fu12fast,jayaram2019diskann,malkov2018efficient}. It maintains two sets, i.e., a search set $\mathcal{S}$ and a result set $\mathcal{R}$, where the size of $\mathcal{S}$ is unbounded to store the candidates yet to be searched and the size of $\mathcal{R}$ is bounded by $N_{ef}$ to maintains the $N_{ef}$ nearest neighbors visited so far. At each iteration, it pops the object with the smallest distance in $\mathcal{S}$ and enumerates its neighbors. For each neighbor, it conducts DCO to check whether its distance to query is no greater than the maximum distance in $\mathcal{R}$. If so, it pushes this object into $\mathcal{R}$ and $\mathcal{S}$, in which the object with maximum distance in $\mathcal{R}$ will be removed whenever $\mathcal{R}$ is full. For a quantization-based method such as \emph{IVF}, it first selects the $N_{probe}$ nearest clusters based on the distance from the query to its centroids. Then, it scans all candidates, in which it maintains a KNN set $\mathcal{K}$ with a max-heap of size $K$~\cite{gao2023high}. For each one, it conducts DCO to check whether its distance to query is no greater than the maximum distance in $\mathcal{K}$. If so, it updates $\mathcal{K}$ with the scanned object. 

\section{Experiment}

\begin{table}[]
\centering
\small
\caption{Dataset Statistics}
\vspace{-0.2cm}
\label{tab:dataset}
\begin{tabular}{c|cccc}
\hline
Dataset  & Cardinality & Dimension & Query Size & Data Type \\ \hline
MSong    & 992,272     & 420       & 200        & Audio     \\
DEEP     & 1,000,000   & 256       & 1,000      & Image     \\
Word2Vec & 1,000,000   & 300       & 1,000      & Text      \\
GIST     & 1,000,000   & 960       & 1,000      & Image     \\
GloVe    & 2,196,017   & 300       & 1,000      & Text      \\
Tiny5M   & 5,000,000   & 384       & 1,000      & Image     \\ \hline
\end{tabular}
\vspace{-0.3cm}
\end{table}

\begin{figure*}
\vspace{-0.4cm}
  \centering
  \includegraphics[width=0.875\linewidth]{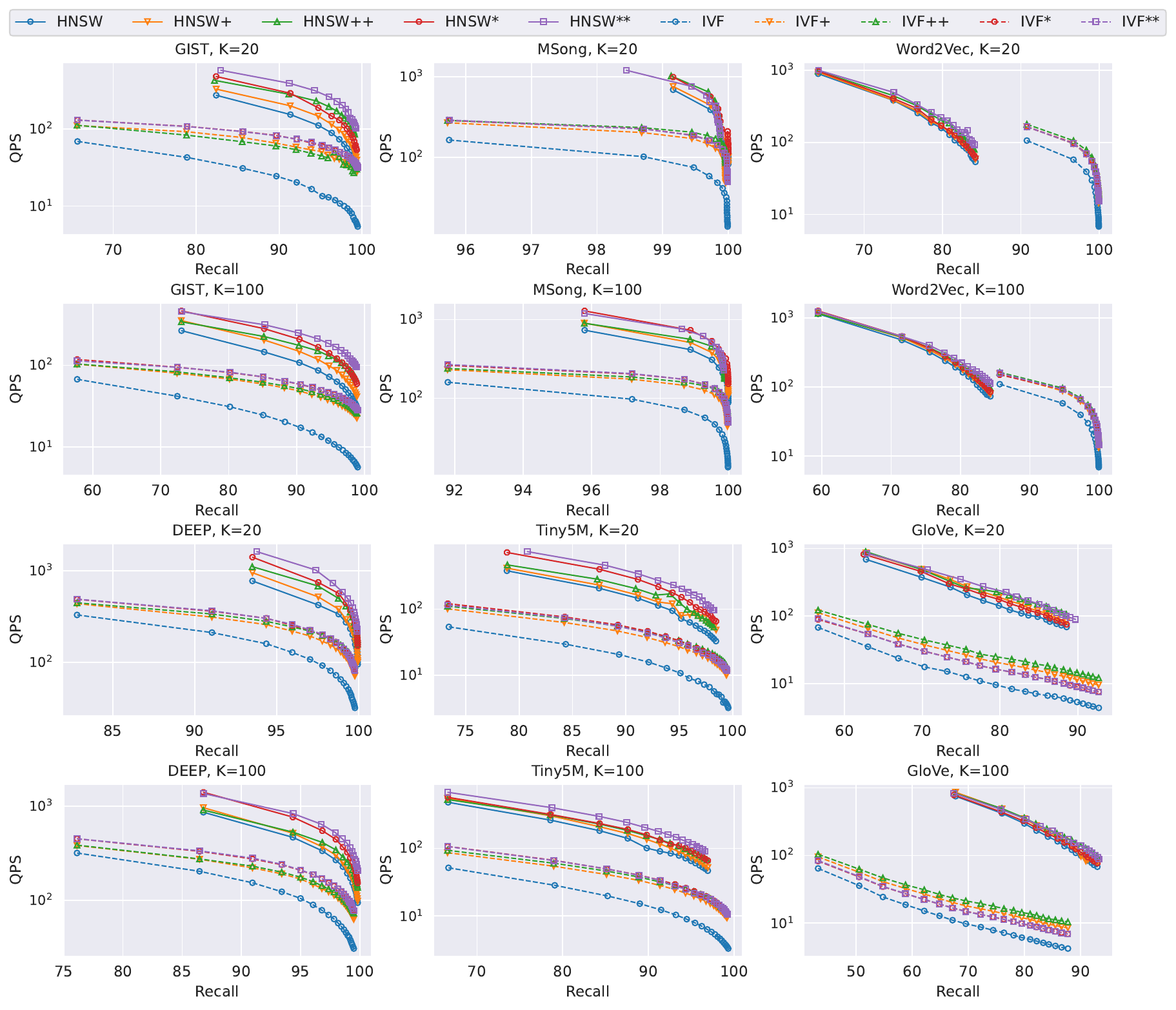}
  \vspace{-0.5cm}
  \caption{Time-Recall Tradeoff.}
  \label{fig:overall_performance}
\vspace{-0.36cm}
\end{figure*}

\subsection{Experimental Settings}
\textbf{Datasets.} 
We conduct our experiments on six public datasets with different cardinalities and dimensionalities to be in line with various benchmark AKNN algorithms~\cite{gao2023high,li2020improving,lu2021hvs}. The dataset statistics are shown in Table~\ref{tab:dataset}.

\noindent \textbf{Algorithms.} 
{\color{black}{For DCOs, We compare our method \emph{DADE} with the conventional method, i.e., $FDScanning$, and the \textbf{SOTA} approach, i.e., $ADSampling$~\cite{gao2023high}. The other distance estimation techniques such as Product Quantization~\cite{jegou2010product} is ignored in our experiments since $ADSampling$ has empirically demonstrate a superior performance compared with them~\cite{gao2023high}.}}
\begin{itemize}[leftmargin=*]
\item $FDScanning$: compute the \emph{exact} distance $dis$ with full $D$ dimensions, and then determine whether $dis < r$. 
\item $ADSamping$: estimate the distance in the low dimension space with \emph{randomly orthogonal transformation}, in which the number of sampled dimensions \emph{adaptively} evaluated, and then determine whether $dis < r$ with the approximated distance $dis'$. 
\end{itemize}
We combine each method of DCOs above with classical AKNN search algorithms, e.g., \emph{IVF} and \emph{HNSW}, and define a set of competitors as follows.
\begin{itemize}[leftmargin=*]
\item \emph{HNSW}~\cite{malkov2018efficient}: the vanilla hierachical navigable small world graph, in which $FDSanning$ is adopted as DCOs.
\item \emph{HNSW+}~\cite{gao2023high}: \emph{HNSW} with $ADSampling$ as DCOs.
\item \emph{HNSW++}~\cite{gao2023high}: \emph{HNSW} with optimizing through decoupling the roles of candidate list, i.e., one for providing the distance threshold for DCOs and one for maintaining the searched objects, in which $ADSamping$ is adopted as DCOs.
\item \emph{HNSW*}: \emph{HNSW} with our proposed \emph{DADE} as DCOs.
\item \emph{HNSW**}: \emph{HNSW++} with our proposed \emph{DADE} as DCOs.
\end{itemize}
Similarly, we define a set of variants of \emph{IVF}, where \emph{IVF++}~\cite{gao2023high} is the cache-friendly version. We refer the readers to the seminal work~\cite{gao2023high} for more details about the variants.

\noindent \textbf{Performance Metrics.}
We use Recall~\cite{zheng2020pm,zhao2023towards,gao2023high,gao2024rabitq}, i.e., the overlap ratio between the results returned by AKNN algorithms and the ground-truths. We adopt query-per-second (QPS)~\cite{lu2021hvs,gao2023high}, i.e., the number of handled queries per second, to measure efficiency. It should be noted that for AKNN search, greater QPS with the same recall indicates a better algorithm. 

\noindent \textbf{Implementations.}
We implement the \emph{HNSW}-related approaches such as \emph{HNSW**} based on hnswlib~\cite{malkov2018efficient} and implement \emph{IVF}-related approaches such as \emph{IVF**} based on Faiss~\cite{johnson2019billion} library. We obtain the transformation matrix for \emph{ADSampling} and \emph{DADE} by using NumPy library. In the query phase, all algorithms are implemented with C++. Following previous studies~\cite{gao2023high,wang2021comprehensive}, all hardware-specific optimizations including SIMD, and multi-threading are prohibited for a fair comparison. {\color{black}{All C++ codes are complied by g++ 7.5.0 with $-O3$ optimization and run in the platform with Ubuntu 16.04 operating system with 48-cores Intel(R) CPU E5-2650 v4 @ 2.20GHz 256GB RAM. }}

\noindent \textbf{Parameter Setting.}
For all \emph{HNSW}-related approaches, two hyper-parameters are empirically preset, i.e., the number of connected neighbors $M$ and the maximum size of the results set $efConstruction$. Following previous studies~\cite{gao2023high}, we set $M$ and $efConstruction$ to $16$ and $500$ respectively. For \emph{IVF}-related approaches, the number of clusters $N_c$ is the critical hyper-parameters in the index phase. Following Faiss~\cite{johnson2019billion}, we set $N_c$ to be around the square root of the cardinality, i.e., $4096$ in our experiments. For $ADSampling$, we use their default parameters, i.e., $\epsilon_0=2.1$. For \emph{DADE}, we empirically set $P_s$ to $0.1$ and the step size of dimension expansion $\Delta_d$ to $32$. 

\begin{figure*}
  \centering
  \includegraphics[width=0.98\linewidth]{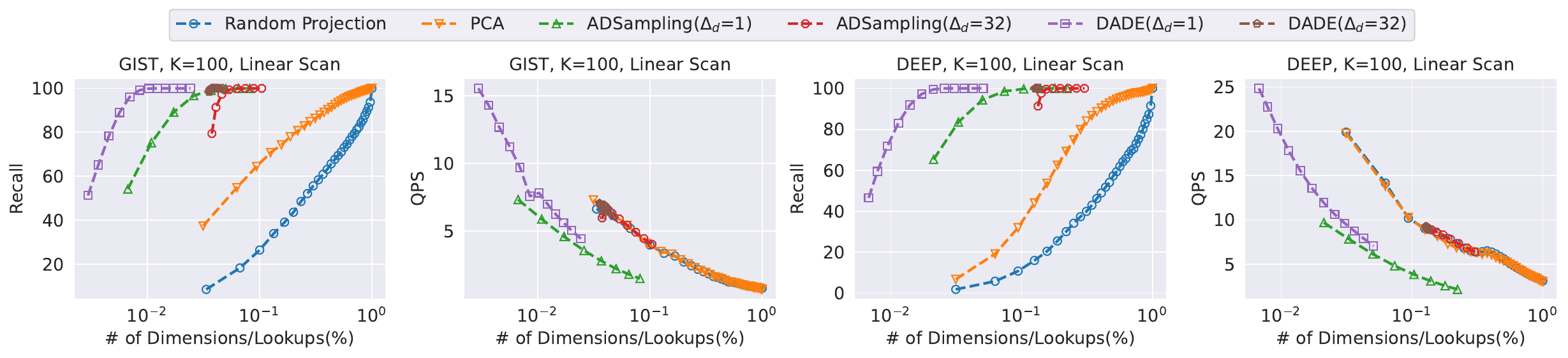}
  \vspace{-0.2cm}
  \caption{{\color{black}{Feasibility for DCOs in terms of Recall and QPS.}}}
  \label{fig:feasibility}
  \vspace{-0.2cm}
\end{figure*}

\begin{figure*}
  \centering
  \includegraphics[width=0.98\linewidth]{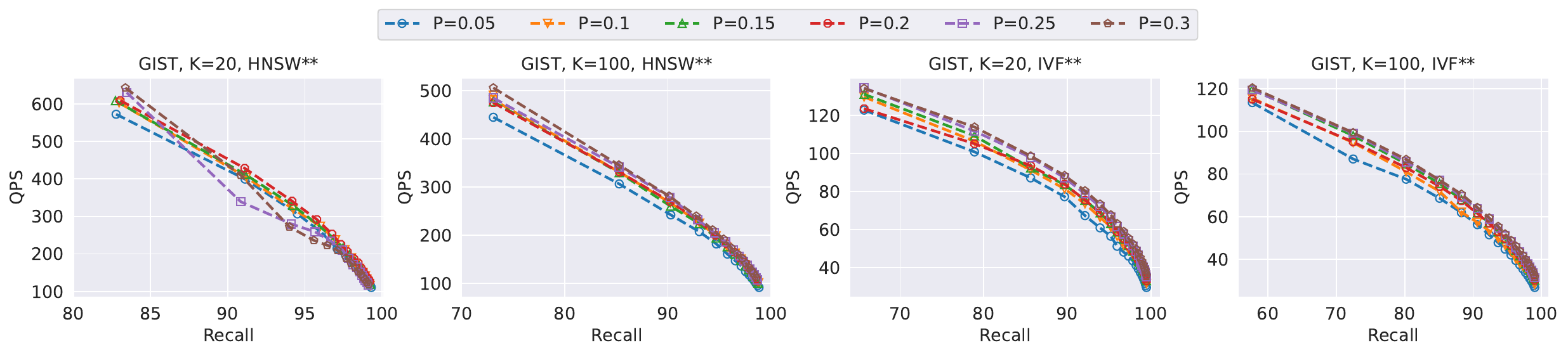}
  \vspace{-0.2cm}
  \caption{{\color{black}{Parameter Study on $p$ of AKNN** Algorithms with Different $K$.}}}
  \label{fig:sensitivity_p}
  \vspace{-0.2cm}
\end{figure*}

\begin{figure*}
  \centering
  \includegraphics[width=0.98\linewidth]{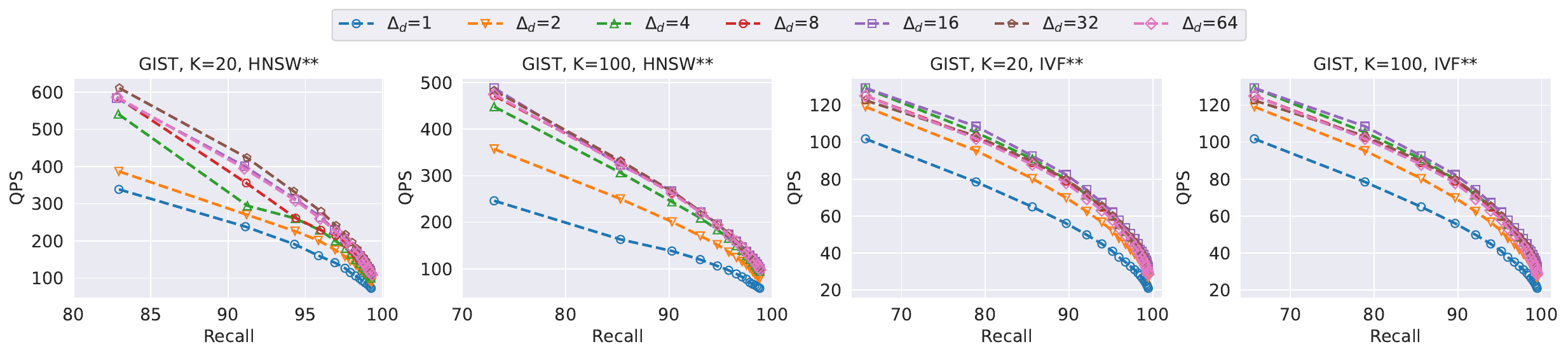}
  \vspace{-0.2cm}
  \caption{{\color{black}{Parameter Study on $\Delta_d$ of AKNN** Algorithms with Different $K$.}}}
  \label{fig:sensitivity_deltad}
\vspace{-0.2cm}
\end{figure*}

\subsection{Experimental Results}

\subsubsection{\textbf{Overall Performance.}} 
\label{sec:overall_performance}
We report the experimental results in terms of Recall and QPS in Figure~\ref{fig:overall_performance}. Specifically, we vary $N_{ef}$ (i.e., the maximum size of the result set $\mathcal{R}$) from $100$ to $1500$ with step size $100$ for the \emph{HNSW}-related approaches, and vary $N_{nprobe}$ (i.e., the number of clusters to be searched) from $20$ to $400$ with step size $20$ for \emph{IVF}-related approaches to show the trade-off between time and accuracy. From the results, we have the following observations. (1) From the index perspective, \emph{HNSW} outperforms \emph{IVF} in most cases. The techniques such as decoupling the roles of the candidate list for \emph{HNSW} and cache-friendly optimization for \emph{IVF} proposed in~\cite{gao2023high} are effective, which clearly improves the efficiency without affecting the accuracy (e.g., \emph{HNSW**} outperform \emph{HNSW*} with a large margin in most cases). (2) From the DCO perspective, \emph{DADE} and $ADSampling$ can achieve better trade-off compared with $FDScanning$ (e.g., \emph{IVF*} and \emph{IVF+} outperform the vanilla \emph{IVF}). Moreover, compared with the best-performing competitor $ADSampling$, our proposed method \emph{DADE} consistently improves the efficiency with a large margin. Recall that the only difference between the methods such as \emph{IVF*} and \emph{IVF+} lies in its method of DCO. For example, on the DEEP dataset with $K=100$, \emph{HNSW*} achieves the recall of $86\%$ with QPS of $140$, while \emph{HNSW+} has little decline of recall with QPS of $96$. In this situation, our method improves the efficiency by more than $45\%$. When focusing on the high accuracy region such as $99\%$ recall on the DEEP dataset, our method even provides a better improvement. For example, $HNSW*$ achieves a recall of $99.2$ with a QPS of $28$, while $HNSW+$ achieves the same recall with a QPS of $18$. 

\subsubsection{\textbf{Feasibility of Distance Estimation Methods for DCOs.}}
\label{sec:feasibility}
Next, we study the feasibility of various distance estimations, i.e., \emph{Random Projection}, \emph{PCA}, \emph{ADSampling}, and \emph{DADE}. To eliminate the effect of different index structures, we follow~\cite{gao2023high} to conduct this experiment with an exact KNN algorithm, called \emph{Linear Scan}. Specifically, we scan all the data objects and return the $K$ nearest neighbors to the queries. For \emph{ADSampling} and \emph{DADE}, we maintain a KNN set like \emph{IVF} and conduct DCOs for each object sequentially. We plot the curves of recall-number of dimensions and QPS-number of dimensions in Figure~\ref{fig:feasibility}, where the value of $x$-axis can be computed through the number of dimensions used for distance calculation in different DCO methods divides the total dimensions when \emph{FDScanning} is adopted. For \emph{random projection} and \emph{PCA}, we vary the dimensionality of the projected vectors (i.e., the estimated distance is based on \emph{equal} dimensions rather than \emph{adaptively} set in query phase). 

From the results, we can observe that (1) when reducing the number of dimensions to $0.1$, the recall of \emph{random projection} will no more than $40\%$; (2) when \emph{PCA} is adopted for distance estimation, the recall has significant improvement compared with \emph{random projeciton}, which demonstrates the necessity to perform data-aware distance estimation; (3) $ADSampling$ and \emph{DADE} only need less than $0.1$ dimensions on average to achieve more than $90\%$ recall, which may demonstrate that for an object that is far away from the query, we only need few dimensions to correctly confirm $dis>r$. Thus, the \emph{exact} distance computation in DCOs is unnecessary; (4) For $ADSampling$ and \emph{DADE}, the curves are plotted through varying $\epsilon_0$ from $0.5$ to $4.0$ and varying $P_s$ from $0.05$ to $0.6$. Compared with $ADSampling$, \emph{DADE} achieves better recall when the number of dimensions equals, which confirms the results from Section~\ref{sec:overall_performance}; (5) compared with \emph{PCA}, \emph{DADE} that adopts hypothesis testing to \emph{adaptively} determine the number of dimensions to approximate the distance, achieves better performance in terms of recall and QPS since different candidates may require different number of dimensions to check whether $dis<r$ with sufficient confidence. This observation also shows the effectiveness of our proposed hypothesis testing method. 

\subsubsection{\textbf{Sensitivity Study.}} 
We study two hyper-parameters in \emph{DADE}, i.e., the significance level $P_s$ in hypothesis testing and the step size $\Delta_d$ in the dimension expansion, where these factors control the trade-off between accuracy and QPS. In the following section, for a given $P_s$ and $\Delta_d$, each curve is plotted through varying $N_{ef}$ and $N_{probe}$ for \emph{HNSW**} and \emph{IVF**}, respectively. 

\noindent \textbf{\underline{The effect of $P_s$.}} We conduct the experiment by varying $P_s$ from $0.05$ to $0.3$ with step size $0.05$. It should be noted that the smaller $P_s$ is, the more accurate the estimated distance is. The results are shown in Figure~\ref{fig:sensitivity_p}. We can observe that the curve moves to the upper right with the increase $P_s$ at the beginning. With the further increase such as from $P_s=0.25$ to $P_s=0.3$ in the GIST dataset with $K=100$, the curves move in an opposite way, which shows a trade-off between the estimated accuracy and efficiency. Specifically, the greater $P_s$ is, the higher probability $H_0$ can be rejected (see Section~\ref{sec:hypothesis}). Thus, \emph{DADE} will exit the loop earlier with a lower number of dimensions, which will improve the efficiency. However, the increases of $P_s$ will cause many more failures (i.e., $H_0$ is rejected with $dis<r$). Therefore, when $P_s$ excels at a specific threshold, the increase in failure will degrade the performance of AKNN search algorithms. 

\noindent \textbf{\underline{The effect of $\Delta_d$.}} We conduct the experiment by varying $\Delta_d$ from $1$ to $64$ with uneven intervals. The results are shown in Figure~\ref{fig:sensitivity_deltad}. From this figure, we have the following observations. (1) With the increase of $\Delta_d$, the curves move from the bottom left to the upper right first. Then, the curves move to the opposite way. This is because when $\Delta_d$ is small such $1$, \emph{DADE} will cost time to increase the number of dimensions to have sufficient confidence to check whether $H_0$ will be rejected. Although smaller $\Delta_d$ has the ability to quit the DCO procedure with lower dimensions, hypothesis testing may be time-consuming. When $\Delta_d$ is large enough such as $32$, \emph{DADE} can increase the number of dimensions with fewer loops, which decreases the times of hypothesis testing and achieves the best trade-off. (2) The search through \emph{Linear Scan} and \emph{HNSW**} (or \emph{IVF**}) shows different preferences. For example, when \emph{Linear Scan} with \emph{DADE} is adopted, $\Delta_d=1$ achieves the best performance (see Figure~\ref{fig:feasibility}). This is because \emph{Linear Scan} treats all data objects as the candidates, in which most of them can be eliminated with high probability when the number of dimensions is lower since the distance between the object and query is far away. 


\section{Related Works}
\noindent \textbf{Approximate K Nearest Neighbor Search.}
{\color{black}{Existing approaches for AKNN search can be roughly divided into (1) graph-based~\cite{zhao2023towards,malkov2018efficient,fu12fast,azizi2023elpis,peng2023efficient}, 
(2) quantization-based~\cite{jegou2010product,ge2013optimized,babenko2014inverted,gao2024rabitq}, (3) tree-based~\cite{beygelzimer2006cover,ciaccia1997m,dasgupta2008random,muja2014scalable} and (4) hashing-based~\cite{datar2004locality,gan2012locality,zheng2020pm,tian2023db}.}}
Different categories of these methods show different advantages. For example, graph-based methods usually achieve the best performance for both in-memory~\cite{zhao2023towards,lu2021hvs} and disk-resident situations~\cite{wang2024starling,jayaram2019diskann,chen2021spann}. Quantization-based methods outperform the others in terms of memory consumption. Hash-based methods provide a theoretical guarantee in terms of the quality of the searched objects. 
Beyond these methods, there are various studies that apply machine learning techniques to AKNN search~\cite{gupta2022bliss,rajput2024recommender,lu2024knowledge,zheng2023learned}. For example, BLISS~\cite{gupta2022bliss} adopts multilayer perceptron to predict the bucket id of each data object and performs repartition to make the objects in each bucket more compact in terms of KNN. 
{\color{black}{Learning to hashing~\cite{wang2017survey} adopts the metric learning technique to group the similar data objects into the same bucket.}}
Zheng et.al~\cite{zheng2023learned} propose to learn the number of buckets to be scanned for each query in \emph{IVF} index to reduce the number of computations in total since different queries require different numbers of buckets to achieve the same accuracy. 
Although these approaches provide different ways to generate the candidates, the methods for DCOs are orthogonal to these approaches, which focus on finding KNNs among the generated candidates. 

\noindent \textbf{Distance Estimation.}
Random Projection is a well-known technique to approximate the Euclidean distance, which is widely used in LSH~\cite{zheng2020pm,tian2023db}. For example, Zheng et.al~\cite{zheng2020pm} propose PM-LSH to perform query-aware hashing for objects, in which the \emph{exact} Euclidean distance $dis$ can be estimated in the projected space with $dis'/\sqrt{m}$, where $m$ is the dimension of the projected space. 
However, it should be noted that the random projection used in LSH is for generating candidates rather than conducting DCOs. 
{\color{black}{Arora et al.~\cite{arora2018hd} propose HD-Index, where the upper-bound of the Euclidean distance is fast estimated through triangular and Ptolemaic inequality to refine candidates. Similarly, Li et al.~\cite{li2017fexipro} adopt transformation with PCA to obtain a tighter upper-bound when inner product distance is applied.}}
Recently, Gao et.al~\cite{gao2023high} propose $ADSampling$ that leverages the \emph{random orthogonal projection} for distance estimation in DCOs, i.e., $\sqrt{D/d}\cdot dis'$. According to our best knowledge, this is the only work that focuses on the speed-up of DCOs. However, the distance estimation in $ADSampling$ is data-oblivious, which hinders the accuracy of the distance estimation for a specific dataset. Moreover, the distance estimation in $ADSampling$ is \emph{unbiased} in terms of the random projection rather than the data distribution since the explicit expression of the data distribution is \emph{unknown}, which makes it hard to develop a theory to provide guarantee about the distance estimation. To bridge this gap, we propose a general formation of the \emph{unbiased} estimation in terms of data distribution and theoretically show an optimized approach for more accurate distance estimation.

\section{Conclusion}
This paper proposes a data-aware distance estimation, called \emph{DADE}, to speed up the process of DCOs, which aims to return KNNs from the candidates set. Specifically, \emph{DADE} first rotates the original space with an \emph{orthogonal transformation} $W$, where the variance of each dimension in the projected space is ordered from large to small. Then, it approximates the \emph{exact} distance in the space with lower dimensions, where the number of dimensions is \emph{adaptively} determined in the query phase through a hypothesis testing approach. Moreover, the probability that controls the significance level is defined and empirically approximated from the data objects in the hypothesis testing. We theoretically prove that the distance estimation in \emph{DADE} is \emph{unbiased} and \emph{optimized} in terms of data distribution if the transformation is \emph{orthogonal}. We conduct extensive experiments on widely used benchmark datasets compared with conventional and SOTA methods for DCOs by combining different index structures. The results demonstrate that our proposed method \emph{DADE} can outperform existing DCO methods to achieve a better trade-off between accuracy and latency. 

\begin{acks}
This work is partially supported by NSFC (No. 62472068), Shenzhen Municipal Science and Technology R\&D Funding Basic Research Program (JCYJ20210324133607021), and Municipal Government of Quzhou under Grant (No. 2023D044), and Key Laboratory of Data Intelligence and Cognitive Computing, Longhua District, Shenzhen.
\end{acks}

\balance
\bibliographystyle{ACM-Reference-Format}
\bibliography{sample}

\end{document}